\documentclass[conference,10pt]{IEEEtran}
\makeatletter
\def\ps@headings{%
\def\@oddhead{\mbox{}\scriptsize\rightmark \hfil \thepage}%
\def\@evenhead{\scriptsize\thepage \hfil \leftmark\mbox{}}%
\def\@oddfoot{}%
\def\@evenfoot{}}
\makeatother
\usepackage{cite}
\usepackage{hyperref}
\hypersetup{
    colorlinks,
    citecolor=black,
    filecolor=black,
    linkcolor=black,
    urlcolor=black
}
\ifCLASSINFOpdf
   \usepackage[pdftex]{graphicx}

\else
   \usepackage[dvips]{graphicx}

\fi
\usepackage[cmex10]{amsmath}
\usepackage{dsfont}
\usepackage{algorithmic}
\usepackage{algorithm}
\newtheorem{theorem}{Theorem}[section]

\newtheorem{pro}[theorem]{Proposition}
\newenvironment{proof}[1][Proof]{%
  \begin{trivlist}
  \item[\hskip \labelsep {\bfseries #1}]}{%
  \end{trivlist}}   %
\newcommand{\qed}{%
  \nobreak \ifvmode \relax \else \ifdim\lastskip<1.5em
  \hskip-\lastskip \hskip1.5em plus0em minus0.5em \fi \nobreak \vrule
  height0.75em width0.5em depth0.25em\fi}
\begin{document}
\title{Single Video Performance Analysis for Video-on-Demand Systems}

\author{\IEEEauthorblockN{James Y. Yang}
\IEEEauthorblockA{Department of ECE\\Coordinated Science Laboratory\\ 
University of Illinois at Urbana-Champaign\\ 
1308 W Main Street\\ Urbana, IL 61801\\ %
    Email: jyyang3@illinois.edu}
\and
\IEEEauthorblockN{Bruce Hajek}
\IEEEauthorblockA{Department of ECE\\Coordinated Science Laboratory\\ 
University of Illinois at Urbana-Champaign\\ 
1308 W Main Street\\ Urbana, IL 61801\\ %
    Email: b-hajek@uiuc.edu}
}


%Telephone: (800) 555--1212, Fax: (888) 555--1212}


\maketitle

\begin{abstract}
We study the content placement problem for cache delivery video-on-demand systems under static random network topologies with fixed heavy-tailed video demand. The performance measure is the amount of server load; we wish to minimize the total download rate for all users from the server and maximize the rate from caches. Our approach reduces the analysis for multiple videos to consideration of decoupled systems with a single video each. For each placement policy, insights gained from the single video analysis carry back to the original multiple video content placement problem. Finally, we propose a hybrid placement technique that achieves near optimal performance with low complexity.
\end{abstract}


% \fi
\IEEEpeerreviewmaketitle

\section{Introduction}

A video-on-demand (VoD) system is an online video delivery system in which peers can request which videos to watch. In order for peers to watch videos without interruptions or large delays, the system must meet stringent delivery requirements - peers need to begin to download quickly and stream at the playback rate.
 
Traditionally, all video requests in a VoD system were handled by a central server. However, as the number of videos and peers grow, it becomes increasingly difficult for one central server to provide all the storage and bandwidth. Moreover, due to the increasing spread of peer locations, the size of the network grows and cost of delivery increases. Therefore, it is reasonable to design a cache delivery system in which each cache acts as a small server to help the central server. Each peer is connected to a subset of the caches and when a peer requests a video to watch, it sends out requests to all of its connected caches. However, if the connected caches do not have the requested video or the rate of upload from caches is insufficient, then the peers turn to the server for the missing parts. Therefore, a reasonable performance measure is the amount of server load.
 
In our analysis, we primarily study the content placement problem: Given a set of caches, what set of videos should be stored at each cache under storage constraints, current demand, and network topology? There are two types of storage methods for this problem: \textbf{1.} Whole storage - videos are stored as whole copies, and \textbf{2.} Fractional storage - with the help of source codes such as maximum-distance-separable (MDS) codes, videos are coded and stored as fractions of a copy. When a peer requests a video, the whole storage architecture requires the peer to download only from one of its connected caches containing the video. In the fractional storage architecture, a peer can simultaneously download from multiple connected caches containing the coded fractions of the video, and the download rates are summed by the additivity property of the MDS code.

Under the aforementioned storage methods, placement policies can be categorized into adaptive or fixed (non-adaptive) placement. Under adaptive placement policies, copies of a particular video are stored in the set of caches with the most received requests. Under fixed placement policies, caches are oblivious to the current demand so random placement of copies of a particular video has the same expected performance as deterministic placement of the copies.
 
Adaptive fractional storage placement minimizes the server load and can be implemented in a distributed way. However, it is inefficient in the sense that it requires overhead and computational power to encode and decode the videos. Optimal adaptive whole storage placement requires is combinatorial in nature and becomes computationally intractable when the system grows large. While fixed fractional storage placement is simple to implement and has linear performance, all caches store some fractions of the same video regardless of the actual connectivity.  Finally, fixed whole storage placement, provides a lower bound on performance. The objective of this work is to compare these four content placement options. We also propose a hybrid placement technique that achieves near optimal performance with less complexity than adaptive fractional storage placement.

This paper is organized as follows. The relationship with related work is described in \autoref{sec:relatedwork}. \autoref{sec:model} describes a model of a VoD cache system with multiple videos. It is a specialization of the model of \cite{HZhang} to the case of no network capacity constraints. \autoref{sec:singlevideoplacement} discusses each of the four placement policies for a single video, and the policies are compared in \autoref{sec:singlevideocomparison}. \autoref{sec:multivideopolicy} gives a general way to use a policy for placement of copies of a single video to produce a policy for placement of multiple videos. In addition, a hybrid policy for multiple video placement is described. \autoref{sec:multivideocomparison} compares the resulting methods for multiple video placement, and the paper is concluded with \autoref{sec:con}.

\section{Related Work} \label{sec:relatedwork}
VoD cache delivery systems have received wide attention for their benefits - the reduction of content delivery cost and the improvement of the end-user performance. Popular video sharing websites such as Youtube have been aggressively deploying cache servers of widely varying sizes at many locations around the world \cite{VKAdhikari}. In addition, cache servers in cache delivery systems appear as television set top boxes or personal computers in peer-to-peer networks like PPLive \cite{PPLive}. There are a number of works on content placement \cite{JAlmeida,YBoufkhad,BTan,XZhou,NLaoutaris,JWu,DApplegate,HZhang}. Almeida et al. \cite{JAlmeida} studied content placement and routing optimization in cache delivery systems subject to path delivery cost under a fixed topology. Boufkhad et al. \cite{YBoufkhad} derived bounds on the number of videos that can be served in the system subject to storage constraint, upload constraint, and cache connection degree. B. Tan et al. \cite{BTan} established an asymptotically optimal content placement strategy subject to a storage constraint and loss network model. Zhou et al. \cite{XZhou} focused on balancing the workload of caches. Laoutaris et al. \cite{NLaoutaris} studied cache storage resource allocation. Since our primary focus is on the content placement problem, the problem is only subject to a given storage constraint and does not involve any path delivery cost or cache upload constraint. In our formulation, we have a fixed topology generated by randomly established cache-peer connections subject to a fixed peer connection degree. As long as any of the caches a peer is connected to stores the requested video, the peer can be served by the cache delivery system without the help of the central server.

Wu and Li \cite{JWu} studied the optimal cache replacement algorithm and suggested that the simplest heuristics perform as well as the optimal algorithms, with very insignificant differences. While our formulation assumes random but static video requests, we also look for a simple suboptimal alternative to the optimal content placement algorithm. In this direction, we decompose the analysis of the content placement problem from the scale of the entire system with multiple videos into decoupled systems, each with only a single video of given popularity. Applegate et al. \cite{DApplegate} formulated the adaptive whole storage placement problem as a mixed integer program (MIP) subject to a storage constraint and link bandwidth constraint. The adaptive whole storage placement problem is solved approximately by MIP. In our work, we derive an upper and a lower bound on the performance of the single video adaptive whole storage placement policy using analytical and heuristic approaches. Zhang \cite{HZhang} used MDS codes to relax the integer constraint and converted the integer program into a convex, adaptive fractional storage placement problem that can be solved exactly. This result provides an upper bound on the performance of any content placement policy, including all single video placement policies. 

With the insights gained from single video placement policy analysis, we return to the original content placement problem by constructing a general method for the four placement policies for multiple videos, which places copies of videos optimally for each policy. Then, motived by observations from single video comparisons, we introduce a hybrid storage multiple video placement policy which is a suboptimal alternative to the adaptive fractional storage placement policy from Zhang \cite{HZhang}. 

To the best of our knowledge, we have been the first to decompose the analysis of the content placement problems for the system with multiple videos into decoupled systems with a single video each.

\section{Model and Assumptions} \label{sec:model}
 
To focus strictly on the analysis of content placement policies in the VoD cache system, we adapt a simple closed homogeneous system model. We consider the server to be external to the cache system and to provide help only when caches within the system cannot satisfy all the peer demand, as shown in \autoref{fig:Cachesystem}. In the model, the numbers of peers, videos, and caches remain fixed. Each peer is connected to an equal fixed number of caches. We assume peer location is uniformly distributed, so peer connections (the set of caches connected to a peer) are also uniformly distributed.

\begin{figure}[htb]
  \begin{center}
\vspace{-.3in}
\hspace{-0.5in}
 \includegraphics[width=0.55\textwidth,height=0.55\textheight, clip]{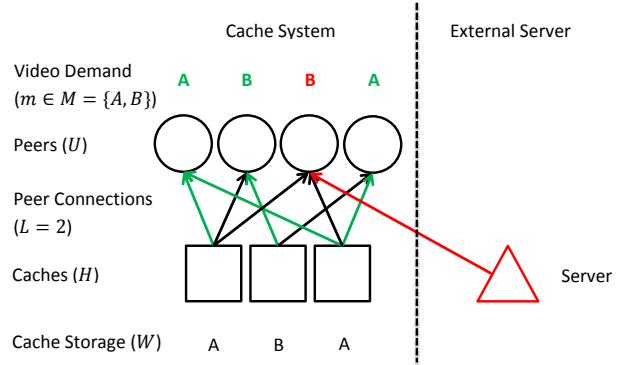}
\vspace{-3in}
    \caption{Cache system of three caches, four peers and two random connections per peer}
\vspace{-0.1in}
    \label{fig:Cachesystem}
  \end{center}
\end{figure}

Under this setting, we have a model of a random but fixed network topology (a graph of cache-peer connections). Each video in the system occupies the same storage space and has the same playback rate. Peers in the system watch videos continuously and independently select videos according to a heavy-tailed popularity distribution, which we assume is a $Zipf$ distribution. Based on the law of large numbers (LLN), the number of peers requesting video $i$ in any instance is close to the expected number. Since the expected number of peers connected to each cache is the same, we assume each cache has sufficient bandwidth to provide the requested streaming rates to all connected peers. We ignore link bandwidth constraints imposed by an underlay network. At any time, peers can only download as much as the required playback rates, so they do not store any future video chunks.

\subsection{Problem Formulation}
The notation of the model is summarized as follows:
 
\begin{itemize}
\item $H$: set of caches
\item $U$: set of peers
\item $M$: set of videos, ordered according to popularity
\item $G$: set of all possible network topology graphs (i.e. peer-cache connection graph)
\item $N^g_u$: set of caches connected to peer $u$ under graph $g\in G$
\item $H_m$: set of caches storing video $m$
\item $U_m$: set of users requesting video $m$
\item $C_m$: number of copies of video $m$ stored in the cache system (possibly sum of fractional parts)
\item $S_h$: storage capacity of cache $h$ in units of videos
\item $Zipf(K,\ 0.8)$: probability distribution over $\{1, ..., K\}$ given by $p(m)=\frac {1/m^{0.8}}{\sum_{n=1}^{K} (1/n^{0.8})}$
\item $p(m)$: video popularity distribution (i.e. probability the video requested by a peer is video $m$)
\item $\alpha(m)$: probability video $m$ is present in a typical cache given by $\alpha(m)=\frac{C_m}{|H|}$
\item $x_{hu}$: download rate of peer $u$ from cache $h$ (possibly a fractional rate)
\item $W_{hm}$: fraction of video $m$ on cache $h$\\
\end{itemize}

Given the set of caches $H$ and the set of peers $U$, under a fixed network topology $g\in G$, each peer $u$ is connected to $|N^g_u|=L$ random caches selected uniformly. Given the set of videos $M$, the probability a peer requests video $m$, $p(m)$, follows the Zipf($|M|$,0.8) distribution, and the mean number of peers requesting video $m$ satisfies $E[|U_m|]=|U|\cdot p(m)$. Each cache $h$ has an equal storage capacity of $S_h$ units of video and stores $W_{hm}$ units of video $m$, where $W_{hm}\in \{0,1\}$ for whole storage placement policies and $W_{hm}\in [0,1]$ for fractional storage placement policies. Each peer $u$ downloads at rate $x_{hu}$ from cache h for each of its connected caches $h \in N^g_u$, where $x_{hu}\in \{0,1\}$ for whole storage placement policies and $x_{hu}\in [0,1]$ for fractional storage placement policies.

The cost function to be minimized is the server load, or equivalently, the utility function to be maximized is the total download rate for all users from the caches. The following optimization problem for a given graph $g=(N^g_u:u\in U)$, was formulated in \cite{HZhang}:
\begin{eqnarray} \label{eqn:optimizationproblem} %
   	& \max & \sum_{u\in U}\min\left\{ \sum_{h\in N^g_u} x_{hu},1\right\} \\
	\nonumber & w.r.t. & x_{hu} \mbox{ and } W_{hm}\  \forall u\in U, h\in H, m\in M\\
	\nonumber & s.t. &  x_{hu} \leq W_{hm} \ \forall h,u\ and\ m:u\in U_m,\\
   	\nonumber & \  & \sum_{m\in M} W_{hm} \leq S_h \ \forall h,\\
	\nonumber & \ & W_{hm} \in\{0,1\} \ \forall h,m \ (\mbox{whole storage placement}), or\\
  	\nonumber & \ & W_{hm} \in[0,1] \ \forall h,m \ (\mbox{fractional storage placement}).\\\nonumber
\end{eqnarray}

\subsection{Fractional Storage and Adaptive Placement Methods}

Fractional storage is a way of storing video copies that allows each copy to be split into pieces of smaller sizes and these pieces are placed in different caches. A fractional storage method is said to have the additivity property, if the useful download rate of a peer is the sum of the download rates from the peer's connected caches. 

One fractional storage method that has the additivity property is source coding using MDS code. Each video chunk is coded into pieces of smaller size and these pieces are placed into caches. Every piece a peer downloads is useful so any subset of the coded pieces that approximately sums to the chunk size can be used to reconstruct the original video chunk. Another fractional storage method that has the additivity property is time-sharing. Each video chunk is split into substreams through time division and these substreams are placed into caches. Only the non-redundant substreams a peer downloads are useful. The pieces of video chunks can also be viewed as video chunks of smaller size, so the placement policy is whole storage again with the whole units being the substreams. Therefore, time-sharing fractional storage methods have more constraints than source coding fractional storage methods. We will use source coding for fractional storage in this paper.

Adaptive placement is a way of placing video copies that maximizes the total download rate for all users from the caches in response to the current demand, i.e. copies of a particular video are stored in the set of caches with the most received requests. Fixed (non-adaptive) placement is placement without knowledge of user requests.

\section{Single Video Placement Policy Analysis} \label{sec:singlevideoplacement}
This section analyzes the four storage method and placement policy pairs discussed in the Introduction, for a single video.

\subsection{Fixed Whole Storage Placement Policy} \label{sec:singlevideoplacementnonadaptivewhole}
For fixed whole storage placement policies, cache content is not placed in according to the actual demand. Given an integer number of video copies, $C$, to be placed in the cache system, a set of caches of cardinality $C$ is randomly selected to store one whole video copy each, and $|H_m|=C$. A peer $u$ can download from at most one of its $L=|N^g_u|$ connected caches that store the requested video. Let $p_{miss}(m)$ denote the probability a peer is not connected to a cache storing video $m$, given by $p_{miss}(m)= \frac{\binom{|H|-C}{L}}{\binom{|H|}{L}}$. Given the set of peers, $U$, requesting video $m$, the expected number of peers that are served by the cache system without help from the server is given by:

\begin{eqnarray} \label{eqn:nonadaptiveatomic} %
	\nonumber && E[\mbox{total download rate provided by caches}] \\ 
	\nonumber & = & E[\mbox{\#\ of peers requesting video $m$ served by caches}] \\
	\nonumber & = & (\mbox{\# of peers requesting video $m$})\cdot \\
	\nonumber & \ & \ \ \ \ \ P\{\mbox{a peer is connected to}\\
	\nonumber & \ & \ \ \ \ \ \ \ \ \ \ \mbox{at least one cache storing video $m$}\}\\
  	\label{eqn:nonadaptiveatomicref} & = & |U|\cdot (1-p_{miss}(m))
\end{eqnarray}

The fixed whole storage placement policy provides a good benchmark for all single video placement policies, because adaptive placement policies and fractional storage placement policies, discussed below, are generalizations of fixed whole storage placement policies. Adaptive placement policies can be made to perform better than fixed placement policies because caches see the actual demand and their video catalogs change accordingly. Also, fixed fractional storage placement policies can be made to perform better than fixed whole storage placement policies, which will be shown in the next subsection, \ref{sec:singlevideoplacementnonadaptivefrac}. 

For purposes of optimization in multiple video systems, we modify the above description to consider the case that the number of copies of video $m$ is a random non-negative integer $X$, with $C=E[X]$. Specifically, let $\alpha(m)$ be the probability video $m$ is present in a typical cache. Then for some $\theta$ with $0\leq\theta<1$, $\alpha(m)\cdot|H|=C=\lfloor C\rfloor+\theta$. Let $X$ be the minimum variance integer valued random variable with mean $\alpha(m)\cdot |H|$. So, $P\{X=\lfloor C\rfloor\}=1-\theta$ and $P\{X=\lfloor C\rfloor+1\}=\theta$. Given $X$, the video is assumed to be placed into a set of caches of cardinality $X$, with all $\binom{|H|}{X}$ possibilities having equal probability. The following proposition provides an upper bound on $p_{miss}(m)$ and therefore a lower bound on the expected total download rate for video $m$: provided by the caches.

\begin{pro}\label{pro:pmiss}   $p_{miss}(m) \leq (1-\alpha(m))^L.$  \end{pro}

\begin{proof}
On one hand, factoring out common terms yields:
\begin{eqnarray}
p_{miss}(m) & = & (1-\theta)\prod^{L-1}_{i=0}\left( \frac{n-k-i}{n-i}  \right)\\
	\nonumber &&~~  +  \theta\prod^{L-1}_{i=0}\left( \frac{n-k-i-1}{n-i}  \right)\\
& = &  \frac{n-k-\theta L}{n}\cdot \prod^{L-1}_{i=1}\frac{(n-k-i)}{(n-i)}  \label{eq.pmiss_bnd1}
\end{eqnarray}
On the other hand, using the convexity of $(a-b\theta)^L$ as a function of $\theta,$ and the fact $f(\theta)\geq  f(0)+\theta f'(0)$ for a convex function $f,$
\begin{eqnarray}
(1-\alpha(m))^L & = &  \left( \frac{n-k-\theta}{n}\right) ^L  \nonumber   \\
& \geq &   \left( \frac{n-k}{n}\right) ^L - \frac{\theta L}{n}\left( \frac{n-k}{n}\right)^{L-1} \nonumber   \\
& = & \left(  \frac{n-k-\theta L}{n} \right)  \left( \frac{n-k}{n}\right)^{L-1}  \label{eq.1minusalpha_bnd1}
\end{eqnarray}
Comparing \eqref{eq.pmiss_bnd1} to \eqref{eq.1minusalpha_bnd1} completes the proof of the proposition.
\qed
\end{proof}

The single video performance of the fixed whole storage placement policy given by (\ref{eqn:nonadaptiveatomicref}) and its lower bound are plotted in \autoref{fig:Nonadaptiveatomic_example} for a video $m$ requested by 20 peers (i.e. $|U|=20$). The expected total download rate provided by the caches is plotted versus the expected number of copies of the video stored a system with 50 caches and four random connections per peer. It can be seen that, due to the randomness in the fixed whole storage placement policy, nearly 3/4 of the caches need to store a copy of the video in order for every peer to be served by the cache system without help from the server. If $L<<|H|$, then whether one of the caches a peer is connected to has the video is approximately independent of whether the other caches connected to have the video. Therefore, $p_{miss}(m) \approx  (1-\alpha(m))^L$, so the two curves on the plot are nearly identical.

\begin{figure}[htb]
  \begin{center}
 \includegraphics[width=0.53\textwidth,height=0.25\textheight, clip]{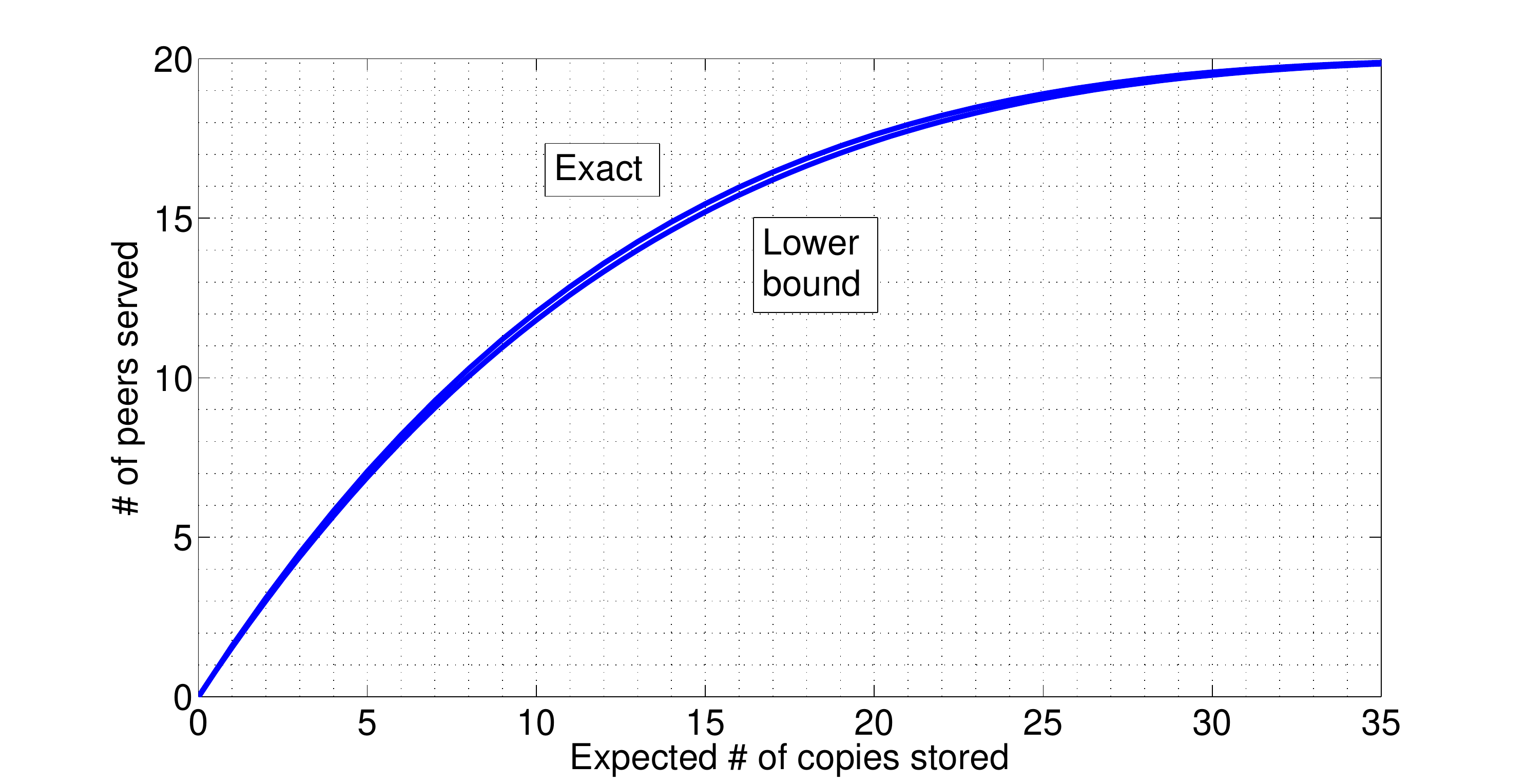}
    \caption{Single video performance of a cache system of 50 caches, 20 peers and four random connections per peer}
    \label{fig:Nonadaptiveatomic_example}
  \end{center}
\vspace{-0.2in}
\end{figure}

\subsection{Fixed Fractional Storage Placement Policy} \label{sec:singlevideoplacementnonadaptivefrac}
For fixed fractional storage placement policies, caches' video catalogs also remain fixed regardless of the demand. For the fixed uniform fractional storage placement policy, with a given number (possibly non-integer) of video copies, $C$, in the cache system, each cache stores the uniform fraction, $W_{h}=\frac{C}{|H|}$, of video $m$. 

\begin{pro}\label{propositionequalfrac}
Among all fixed fractional storage placements, the uniform fractional storage placement maximizes the expected total download rate.
\begin{proof}
It suffices to show the expected download rate of a peer requesting video $m$ served by $L$ random connections is maximized by the uniform fractional storage placement.
 
The download rate of peer $u$ from cache $h$ is given by $ x_{hu} =  W_{hm}\cdot\mathds{1}{\left\{ h\in N_u\right\}}$. Therefore,
\begin{eqnarray}
	\nonumber &&E[\mbox{download rate of peer $u$ from caches}] \\
	\nonumber & = & E\left[\min\left\{\sum_{h\in N^g_u}x_{hu},1\right\}\right]\\
  	\label{eqn:jensenineq} & \leq & \min\left\{\sum_{h\in N^g_u}E[x_{hu}],1\right\} \mbox{ by Jensen's inequality}\\
	\nonumber & = & \min\left\{\sum_{h\in N^g_u}W_{h}\cdot P\left\{h\in N^g_u\right\},1\right\}\\
	\nonumber & = & \min\left\{\sum_{h\in N^g_u}W_{h}\cdot \frac{L}{|H|},1\right\}\\
	& = & \min\left\{L\frac{C}{|H|},1\right\}
\end{eqnarray}
Equality in (\ref{eqn:jensenineq}) is achieved if $\sum_{h\in N^g_u}x_{hu}\equiv\sum_{h\in N^g_u}E[x_{hu}]$, which is true for uniform fractional storage placement.
\qed
\end{proof}
\end{pro}

A peer can download from all of its $L=|N^g_u|$ connected caches and all peers download at the same rate. However, peers may download at a fraction of the required playback rate, so possibly none of them are fully served by the caches. Summing the download rates from all peers requesting video $m$, we find the total download rate provided by the cache system for the set of peers, $U$, without help from the server:
\begin{eqnarray} \label{eqn:nonadaptivefractional} %
	\nonumber && (\mbox{total download rate provided by caches}) \\
	\nonumber & = & \sum_{u\in U}\min\left\{\sum_{h\in N^g_u} \frac{C}{|H|},1\right\}\\
  	& = & |U|\cdot \min\left\{L\frac{C}{|H|},1\right\}
\end{eqnarray}

Given the number of copies of video $m$, $C$, we can find the expected fraction of total download rate provided by the cache system without help from the server. Since caches are uniformly randomly selected by peers, the expected download rate served by each connection is the same for any fixed fractional storage placement. On one hand, the download rate served by each connection under the fixed uniform fractional storage placement is deterministic. As a result, for each additional video copy placed in the cache system, a peer's download rate increases linearly and deterministically. On the other hand, the download rate served by each connection under other fixed fractional storage placements is random. Therefore, as shown in the following proposition, fixed uniform fractional storage placement outperforms any other fixed fractional storage placements, in particular the fixed whole storage placement.

Note that the whole storage placement is a special case of fractional storage placement where fractions are constrained to be zero or one. The performance of the fixed uniform fractional storage placement policy is plotted in \autoref{fig:Nonadaptivefluid_example} for a video $m$ requested by 20 peers (i.e. $|U|=20$). The total download rate provided by the caches is plotted versus the number of copies of the video stored in the cache system. Due to the linear and deterministic properties of the fixed fractional storage placement policy, the marginal performance gain of an additional video copy is always the same until the number of copies reaches the minimum number needed to serve the entire video to all peers. Every cache only needs to store a fraction, $\frac{1}{L}$, of the video in order for every peer to be served by the cache system without help from the server.

\begin{figure}[htb]
  \begin{center}
 \includegraphics[width=0.53\textwidth,height=0.25\textheight, clip]{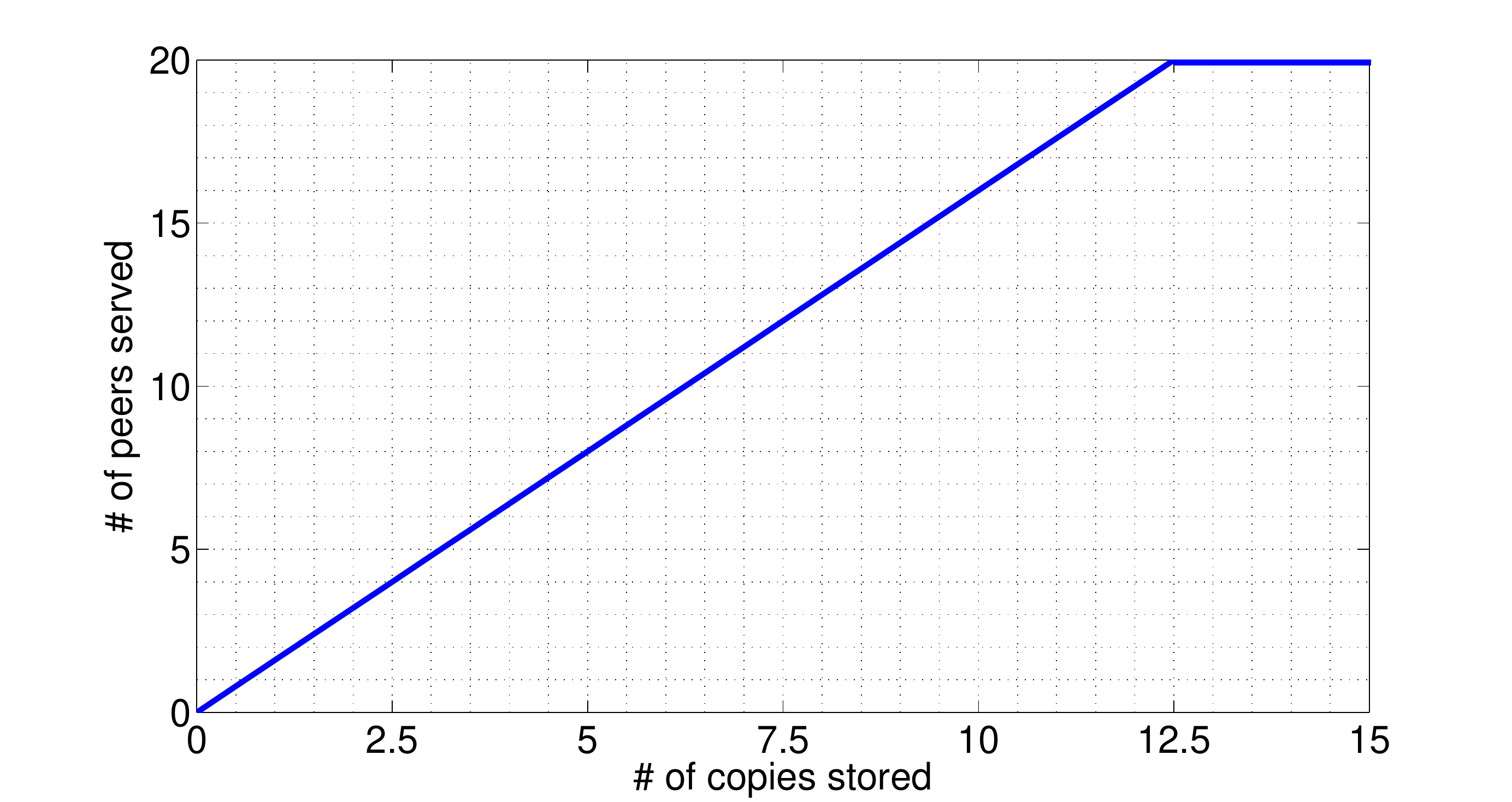}
    \caption{Single video performance of a cache system of 50 caches, 20 peers and four random connections per peer}
    \label{fig:Nonadaptivefluid_example}
  \end{center}
\vspace{-0.2in}
\end{figure}

\subsection{Adaptive Whole Storage Placement Policy}
For adaptive whole storage placement policies, caches select their video catalogs in response to the actual demand. For optimal adaptive whole storage placement, given an integer number of video copies, $C$, to be placed in the cache system, the set of caches of cardinality $C$ that is connected to the largest group of peers requesting video $m$ stores the video; we denote this set of caches by $H_m$. The total download rate provided by the cache system for the set of peers, $U$, without help from the server is the number of peers requesting video $m$ that are connected to $H_m$. Although finding $H_m$ is an NP-hard problem (because the set covering problem is NP-complete), but in this context near optimal performance is provided by a greedy algorithm.

Let $p_{miss}(m)$ denote the probability a peer is not connected to a cache storing video $m$, given by $p_{miss}(m)= \frac{\binom{|H|-C}{L}}{\binom{|H|}{L}}$. We have the following upper bound:

\begin{pro}\label{pro:propositionupbd}
$E$[total download rate provided by $H_m$] $\leq \sum_{\tau=1}^{|U|}\min\left(\binom{|H|}{C}\cdot \mbox{Bin}^c\left(|U|, 1-p_{miss}(m),\tau-1\right),1\right)$, where $Bin^c(N,p,K)$ is the complementary CDF of the binomial distribution at $K$ with corresponding number of trials $N$, and probability of success for each trial $p$. i.e. $Bin^c(N,p,K)=1-Bin(N,p,K)$. The index $\tau$ ranges over the positive integers less than or equal to $|U|$, the number of peers requesting video $m$.
\begin{proof} 
Consider any fixed set $A$ of cardinality $C$,
\begin{eqnarray} \label{eqn:connectedprob1} %
	\nonumber &&P\{\mbox{a given peer is connected}\\
	\nonumber	&\ & \ \ \ \ \ \  \mbox{to at least one cache in $A$}\}\\
	&=&1-p_{miss}(m)
\end{eqnarray}
So, the number of peers connected to at least one cache in A has the binomial distribution with parameters $N=|U|$ and $p$ from (\ref{eqn:connectedprob1}). Thus, for any integer $\tau \geq 1$,
\begin{eqnarray} \label{eqn:adaptivewholeprob2} %
	\nonumber && P\{\mbox{$A$ covers at least $\tau$ peers}\}\\
	&=& Bin^c(|U|,1-p_{miss}(m),\tau-1)
\end{eqnarray}
Let $Y$ be the number of sets of caches of cardinality $C$ that have at least $\tau$ peers connected to them. Since there are $\binom{|H|}{C}$ sets of caches of cardinality $C$ and any such set of caches has probability $p$ to cover at least $\tau$ peers as $A$,
\begin{eqnarray} \label{eqn:adaptivewholetao} %
	\  E[Y] & = & \binom{|H|}{C}\cdot P\{\mbox{$A$ covers at least $\tau$ peers}\}\\
	& = & \binom{|H|}{C}\cdot Bin^c(|U|,1-p_{miss}(m),\tau-1)\label{eqn:adaptivewholetaosub}
\end{eqnarray}
By the first moment bound, for the non-negative integer random variable $Y$,
\begin{eqnarray} \label{eqn:adaptivewholebd} %
	\nonumber && P\{\mbox{\# of peers served by $H_m$ caches}\geq \tau\}\\
	 & = & P\{Y \geq 1\} \ \ \leq \ \ \min\{E[Y],1\}
\end{eqnarray}
Finally, 
\begin{eqnarray}\label{eqn:expectation}
	 \nonumber &&E[\mbox{total download rate provided by the $H_m$}]\\
	 & = &  \sum_{\tau=1}^{\tau=|U|}P\{\mbox{\# of peers served by $H_m$ caches}\geq \tau\}
\end{eqnarray}
Therefore, (\ref{eqn:adaptivewholetaosub})-(\ref{eqn:expectation}) yields the proposition
\qed
\end{proof}
\end{pro}

Next, we introduce a heuristic method for obtaining a suboptimal set of caches with cardinality $C$. In each iteration, the algorithm first finds a cache that is connected to the largest number of peers requesting video $m$ and stores a copy in that cache. Then, it peels away (removes) all peers connected to that cache requesting video $m$. The algorithm is shown in \autoref{alg:adaptivewholeyalg}.

\begin{algorithm}
  \caption{Greedy peeling algorithm for adaptive whole storage placement of C copies of a single video m}
  \label{alg:adaptivewholeyalg}
\begin{algorithmic}[1]
  \WHILE {Not all $C$ copies of video $m$ are placed in the caches}
  \STATE - Find $h_{max}$ with the most connected peers requesting video $m$
  \STATE - Store a copy of video $m$ in cache $h_{max}$
  \STATE - Remove all peers connected to $h_{max}$ requesting video $m$ from the cache system
  \ENDWHILE
  \end{algorithmic}
\end{algorithm}

While not optimal, the greedy peeling algorithm provides a lower bound on the number of peers connected to $H_m$ for any graph. Therefore, the statistical average of total download rate provided by sets of caches of cardinality $C$ chosen by the greedy peeling algorithm over random graphs is a lower bound on the expected total download rate provided by the cache system without help from server.

The performance of the adaptive whole storage placement policy is plotted in \autoref{fig:Adaptiveatomic_example} for a video $m$ requested by 20 peers (i.e. $|U_m|=20$), where the upper bound is obtained from Proposition \ref{pro:propositionupbd} and the lower bound is obtained from \autoref{alg:adaptivewholeyalg}. The bounds on the total download rate provided by the caches are plotted versus the number of copies of the video stored in the cache system. Note that the upper bound is obtained analytically and the lower bound is obtained as an average over random graphs computed by simulation. It can be seen that only a small portion of the caches need to store a copy of the video in order for every peer to be served by the cache system without help from the server. The marginal performance gain decreases as more copies are added because some caches have more peers connected than others and once a peer is covered by one copy it is not counted again for additional copies.

\begin{figure}[htb]
  \begin{center}
 \includegraphics[width=0.53\textwidth,height=0.25\textheight, clip]{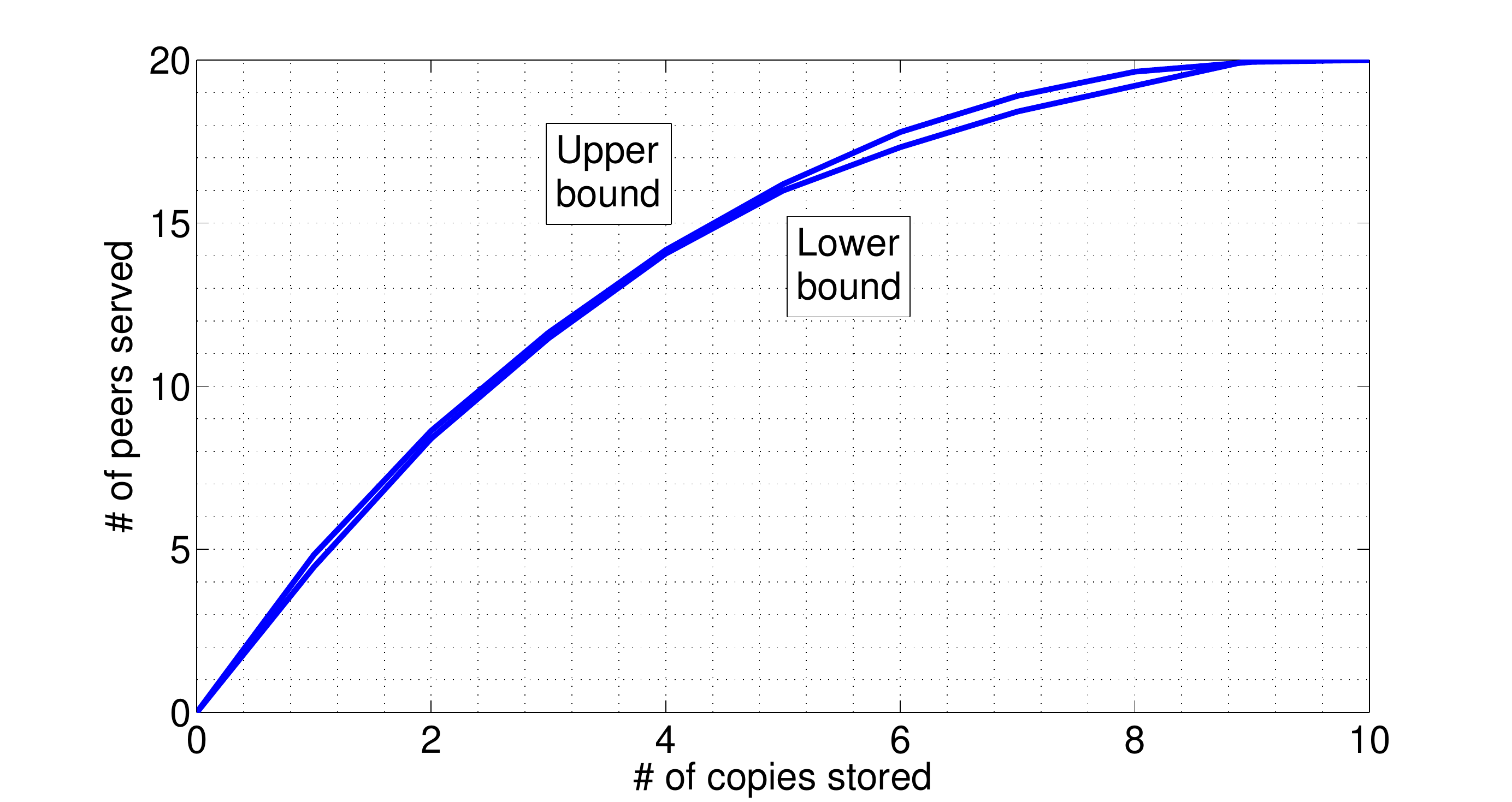}
    \caption{Single video performance of a cache system of 50 caches, 20 peers and four random connections per peer}
    \label{fig:Adaptiveatomic_example}
  \end{center}
\vspace{-0.2in}
\end{figure}

\subsection{Adaptive Fractional Storage Placement Policy}\label{sec:adapfrac}
For adaptive fractional storage placement policies, caches' video catalogs are also selected in response to the actual demand. Given the number (possibly non-integer) of video copies $C_m$ to be placed in the cache system, each cache stores some fraction of video $m$, $W_{hm}$. The values ($W_{hm}: h\in H, m\in M$) are chosen to maximize the total download rate provided by the cache system without help from the server. Since adaptive fractional storage placement relaxes the integer constraint in (\ref{eqn:optimizationproblem}), we can solve the convex optimization problem exactly. A primal-dual algorithm can be applied for solving the convex optimization problem. We first set the download rate of each cache-peer connection and the storage of video $m$ in each cache as primal variables. We then set the price of download rate of each cache-peer connection subject to the availability of video $m$ in the cache and the price of total cache storage of video $m$ subject to a global cache storage constraint as dual variables. Thus the cache system's fractional storage placement, routing, and upload rates converges optimally by the primal-dual algorithm, shown as \autoref{alg:adaptivefractionalalg}, which is a striped down variation of the algorithm in \cite{HZhang}. The notation ``$[a]^+_{x_{hu}}$" on the right-hand-side of \eqref{eqn:adaptivefractionaldownloadrate} denotes $a$ if $x_{hu}>0$ and $\max{\{a,0\}}$ if $x_{hu}\leq 0$.

\begin{algorithm}
  \caption{Primal-dual algorithm for adaptive fractional storage placement of $C$ copies of a single video $m$}
  \label{alg:adaptivefractionalalg}
\begin{algorithmic}
  \STATE Primal 1: update the download rates\\
    \begin{eqnarray} \label{eqn:adaptivefractionaldownloadrate} %
	\dot{x}_{hu}=\left[\delta_{hu}\cdot\left(\mathds{1}{\left\{\sum_{h'\in N^g_u} x_{h'u}<1\right\}}-\lambda_{hu}\right)\right]^+_{x_{hu}}
    \end{eqnarray}
    where $\delta_{hu}>0$ is an adaptation parameter
  \STATE Dual 1: update the price of download rates
    \begin{eqnarray} \label{eqn:adaptivefractionaldownloadprice} %
	\dot{\lambda}_{hu}=\left[\kappa_{hu}\cdot\left(x_{hu}-W_{h}\right)\right]^+_{\lambda_{hu}}
    \end{eqnarray}
    where $\kappa_{hu}>0$
  \STATE Primal 2: update the fractions of the video stored
    \begin{eqnarray} \label{eqn:adaptivefractionalstorage} %
	\dot{W}_{h}=\left[\iota_{h}\cdot\left(\left(\sum_{u:h\in N^g_u}\lambda_{hu}\right)-\omega\right)\right]^+_{W_{h}}
    \end{eqnarray}
    where $\iota_{h}>0$
  \STATE Dual 2: update the price of cache storage
    \begin{eqnarray} \label{eqn:adaptivefractionalstorage} %
	\dot{\omega}=\left[\nu\cdot\left(\sum_{h\in H}W_{h}-C\right)\right]^+_{\omega}
    \end{eqnarray}
    where $\nu_{h}>0$
  \end{algorithmic}
\end{algorithm}

The performance of the adaptive fractional storage placement policy is plotted in \autoref{fig:Adaptivefluid_example} for a video $m$ requested by 20 peers (i.e. $|U_m|=20$). The total download rate provided by the caches is plotted versus the number of copies of the video stored in the cache system. Only a small portion of the caches need to store a copy of the video in order for every peer to be served by the cache system without help from the server. The concavity in the non-decreasing marginal performance gain follows from the random connections which result in a non-symmetric graph and the fact a user has no use for download rate in excess of one. The slower decrease in the marginal performance gain of the adaptive fractional storage placement policy compared to the adaptive whole storage placement policy resulted from the relaxation of the integer storage placement constraint.

\begin{figure}[htb]
  \begin{center}
 \includegraphics[width=0.53\textwidth,height=0.25\textheight, clip]{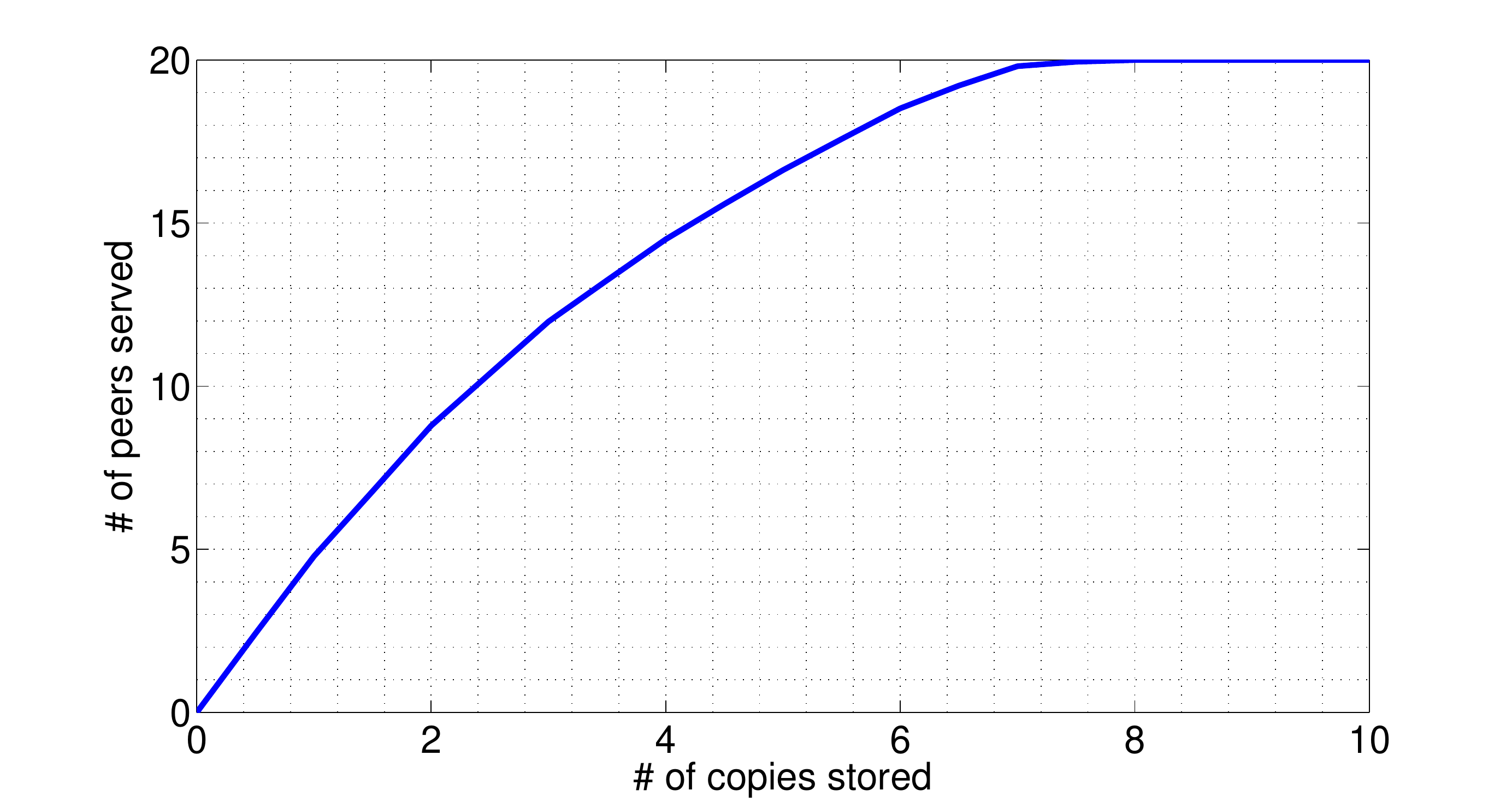}
    \caption{Single video performance of a cache system of 50 caches, 20 peers and four random connections per peer}
    \label{fig:Adaptivefluid_example}
  \end{center}
\vspace{-0.2in}
\end{figure}

\section{\bf Single Video Placement Policy Comparisons}\label{sec:singlevideocomparison}

This section explores the trade-offs between performance and practicality of the placement policies for a single video discussed in \autoref{sec:singlevideoplacement}. To look for potential patterns in single video placement performances, we choose 20, 100, and 2000 peers representing three different popularity levels, where each peer is connected to four random caches out of the 50 caches selected uniformly. Videos with 20 or fewer peer requests represent the set of videos with below average popularity. A video with 100 peer requests represents a video of above average popularity. A video with 2000 peer requests represents the most demanded video. The performance of single video placement policies derived from the analysis in \autoref{sec:singlevideoplacement} are shown in Figures \ref{fig:Singlevideo_example1} - \ref{fig:Singlevideo_example4}. The fraction of total download rate provided by the caches is plotted versus the number of copies of the video stored in the cache system. The abbreviation FW, FF, AW, and AF represents the fixed whole, fixed fractional, adaptive whole, and adaptive fractional storage placement policy, respectively.

\begin{figure}[htb]
  \begin{center}
 \includegraphics[width=0.53\textwidth,height=0.25\textheight, clip]{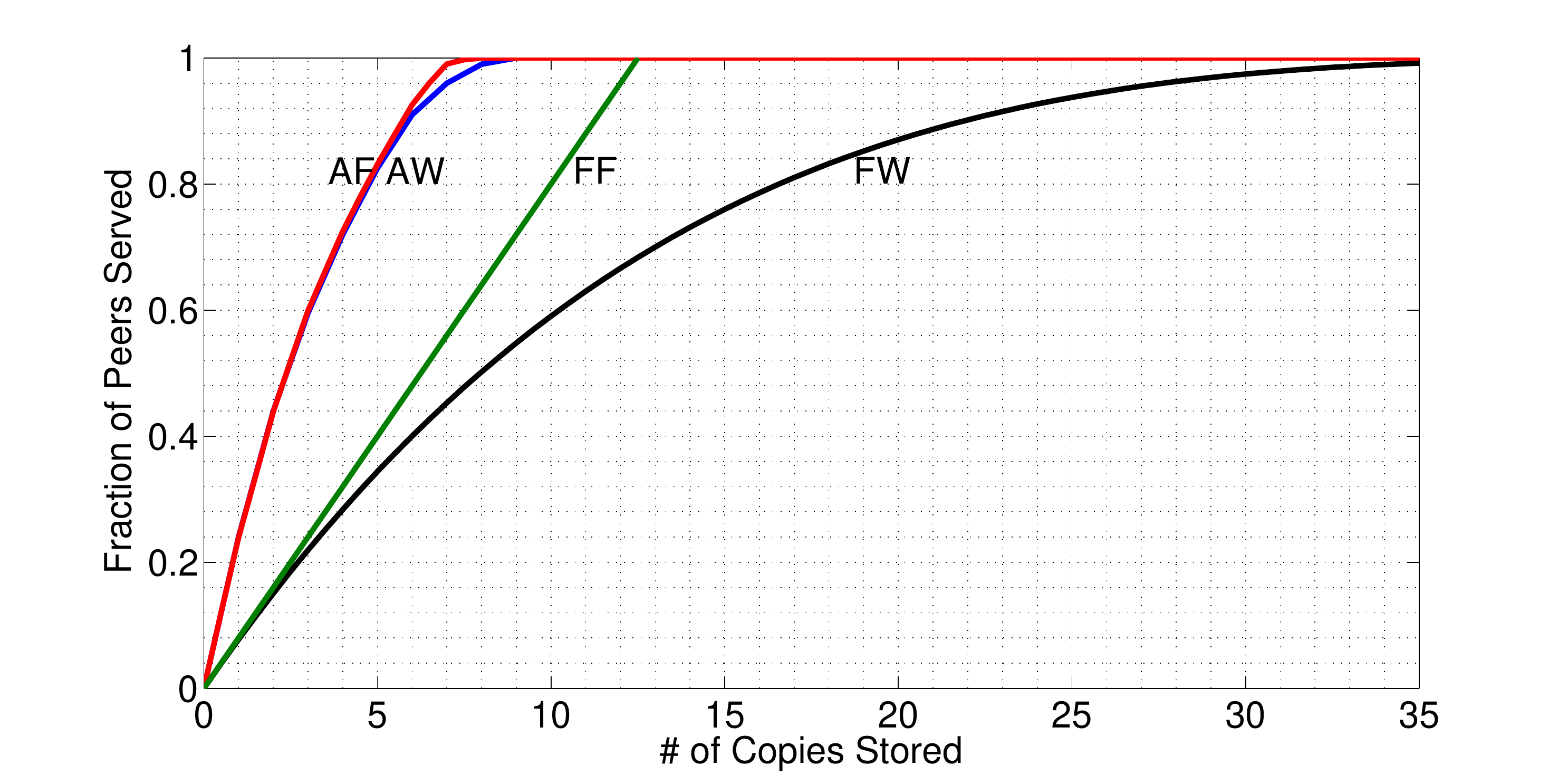}
    \caption{Single video performance of a cache system of 50 caches, 20 peers and four random connections per peer}
    \label{fig:Singlevideo_example1}
  \end{center}
\vspace{-0.2in}
  \begin{center}
 \includegraphics[width=0.53\textwidth,height=0.25\textheight, clip]{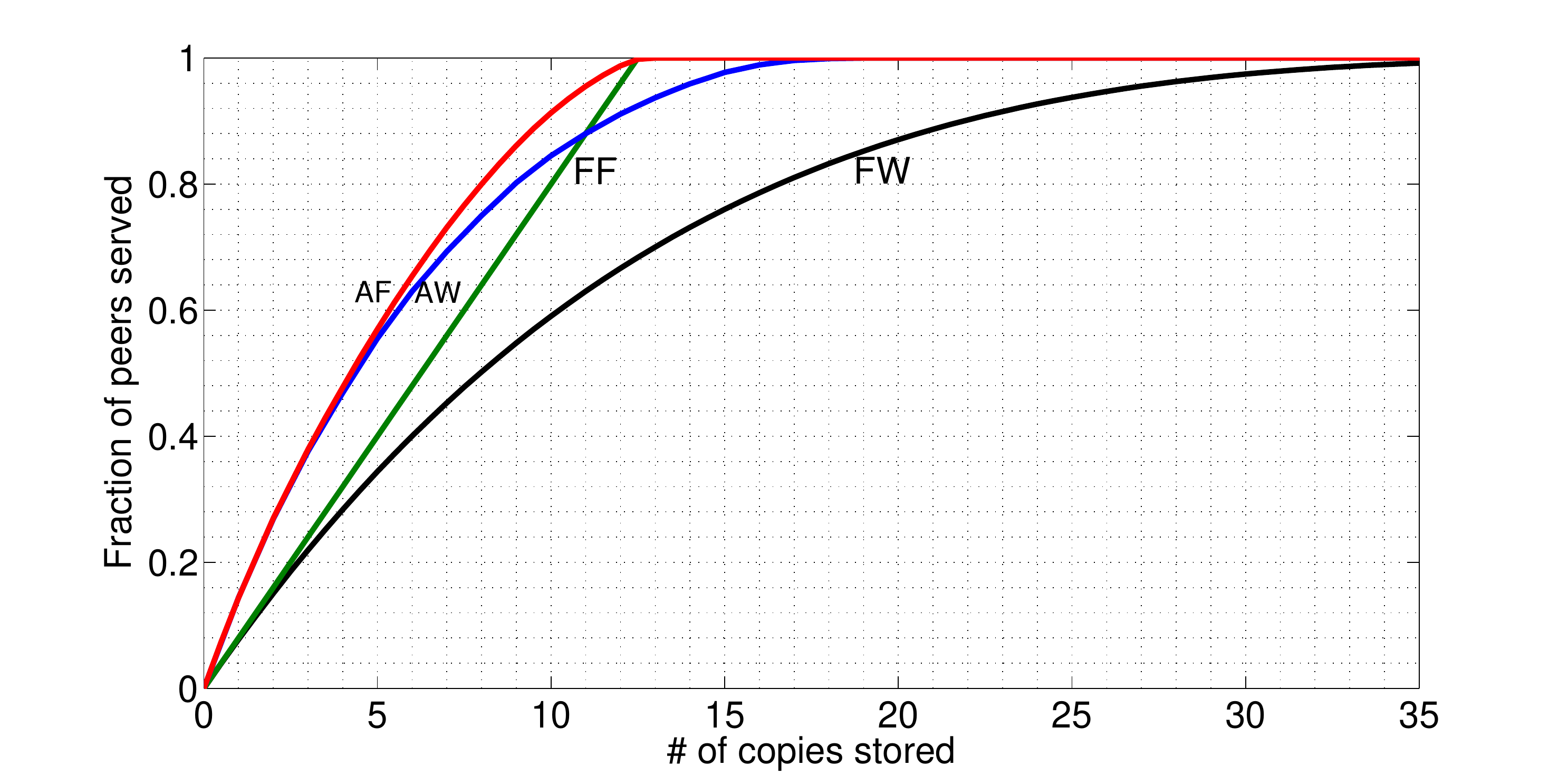}
    \caption{Single video performance of a cache system of 50 caches, 100 peers and four random connections per peer}
    \label{fig:Singlevideo_example3}
  \end{center}
\vspace{-0.2in}
\end{figure}
\begin{figure}[htb]
  \begin{center}
 \includegraphics[width=0.53\textwidth,height=0.25\textheight, clip]{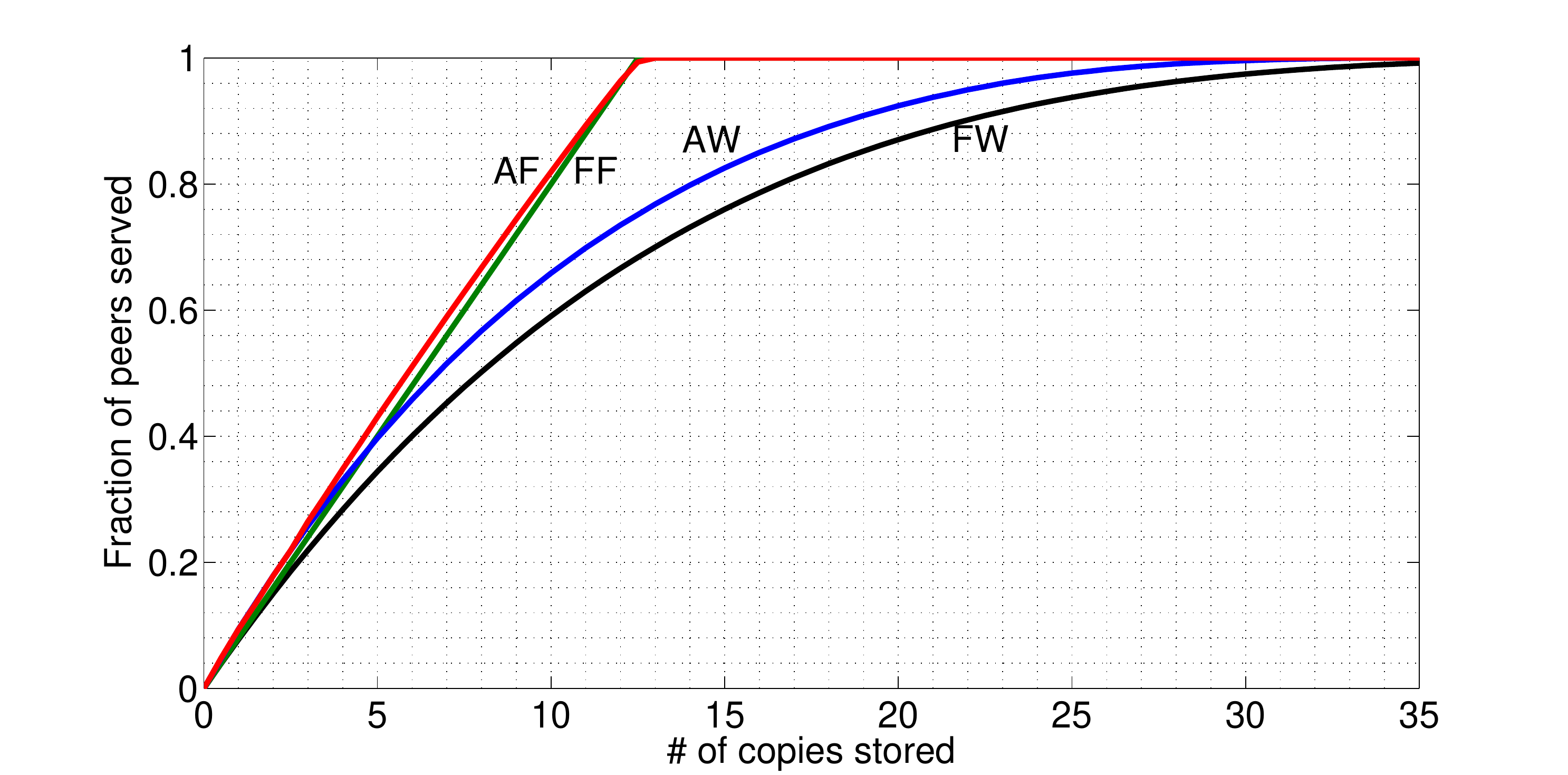}
    \caption{Single video performance of a cache system of 50 caches, 2000 peers and four random connections per peer}
    \label{fig:Singlevideo_example4}
  \end{center}
\vspace{-0.2in}
\end{figure}

One observation from these three plots is that the single video placement service curves of fixed whole and fixed fractional storage placement policies remain the same across videos of all popularity. This is true because the non-adaptive property of these two policies yields the same expected performance over random graphs. 

A second observation is that the fixed whole storage placement policy serves as the performance lower bound and the adaptive fractional storage placement policy serves as the performance upper bound on all four single video placement policies. This fact is true for videos of all popularity because of the integer constraint of whole storage placement policies and the non-adaptive property of fixed placement policies. A consequence of this fact is that the service curve of adaptive whole storage placement policy is always above that of the fixed whole storage placement policy, and the service curve of adaptive fractional storage placement policy is always above that of the fixed fractional storage placement policy. The adaptive whole storage and fixed fractional placement policies are the only pair of curves that can cross.

A third observation is that the service curves of adaptive whole and adaptive fractional storage placement policies drop down and converge to those of their corresponding fixed placement policies as videos become more popular. This is due to the fact that as more peers request a video, the LLN implies each cache is approximately connected to the same number of peers. As a result, the adaptive property of placement policies becomes less beneficial, which means the placement of copies of a video into any random set of caches yields approximately the same performance.

A fourth observation is that the service curve of the adaptive whole storage placement policy rises up and converges to that of the adaptive fractional storage placement policy as videos become less popular. This is due to the fact that as fewer peers request a video, there will be more imbalance in the cache-peer connections due to randomness. As a result, the integer constraint relaxation of the adaptive fractional storage placement policy becomes less imposing, which means the adaptive placement of whole copies of a video yields approximately the same performance.

\section{\bf Multiple Video Placement Policy}\label{sec:multivideopolicy}
In this section, we first describe a general method to apply the algorithms and analysis developed for single video placement to placement of multiple videos. The general method optimally allocates the numbers of copies of the videos by maximizing the total download rate provided by the cache system without help from the server. Then, we examine how the general method works for each of the four storage/placement methods discussed in \autoref{sec:singlevideocomparison}. Finally, motivated by observations from single video comparisons, we introduce a hybrid storage multiple video placement policy. 

The system model remains the same as the one described in \autoref{sec:model}. Given the storage capacity of each cache, the goal is to maximize the total download rate for all users from the caches. Since we are considering the placement of copies of different videos in the multiple video scenario, to simplify the problem, we will consider a single combined total cache storage constraint first. Then, we will add the uniform per-cache storage constraint and comment on the difference.

\subsection{General Algorithm for Multiple Video Placement Policy}
Our general method for using a policy for single video placement for the problem of multiple video placement is summarized as \autoref{alg:greedywhole}. It iteratively stores copies of videos with the highest marginal performance gain subject to the storage constraint. This method maximizes the total expected download rate provided by the cache system without help from the server. This is because placing copies of one video will not affect the marginal performance gains of other videos and the marginal performance gain of a video depends only on the numbers of copies already placed for the video. We will apply this general method to all multiple video placement policies in this section to obtain the optimal video placement subject to some storage constraint.

\begin{algorithm}
  \caption{General algorithm for multiple video placement}
  \label{alg:greedywhole}
\begin{algorithmic}[1]
  \STATE Perform single video analysis on each video to obtain the set of placement service curves for all videos
  \WHILE {Storage constraint not violated}
  \STATE - Find a video with the highest marginal performance gain given copies of videos previously stored
  \STATE - Store a copy of the video in the cache position obtained from the single video analysis
  \ENDWHILE
  \end{algorithmic}
\end{algorithm}

\subsection{Fixed Whole Storage Multiple Video Placement Policy}
Recall from the single video analysis for fixed whole storage placement policy that $\alpha(m)$ is the probability video $m$ is present in a typical cache. In this context, \autoref{alg:greedywhole} reduces to finding the $\alpha(m)$ to minimize the miss probability. 

This way of selecting $\alpha(m)$ is nonadaptive in the sense that the particular set of caches in which the video is placed is independent of which caches the peers are connected to, but adaptive in the sense it relies on the popularity distribution. Recall the set of videos, $M$, follows a decreasing Zipf popularity distribution, $p(m)$ for each $m\in M$, and $p_{miss}(m)$ is the probability a peer is not connected to a cache storing video $m$, given by $p_{miss}(m)= \frac{\binom{|H|-C_m}{L}}{\binom{|H|}{L}}$.

Now consider multiple videos. Proposition \ref{pro:pmiss} yields the following upper bound on $p_{miss}$, the average probability a peer is not connected to a cache storing the requested video:
\begin{eqnarray} \label{eqn:alphamin} %
	p_{miss} \stackrel{\triangle}{=}  \sum_{m=1}^{|M|}  p(m)\cdot p_{miss}(m) \leq  \sum_{m=1}^{|M|}  p(m)\cdot (1-\alpha(m))^L
\end{eqnarray}
Using the upper bound on $p_{miss}$, \autoref{alg:greedywhole} can be carried out analytically. Suppose each cache can store $K$ copies of videos. To minimize the upper bound on $p_{miss},$  we select $(\alpha(m): m \in M )$ to minimize $\sum_{m\in M} p(m)(1-\alpha(m))^L$ subject to the constraint $\sum_{m=1}^{|M|} \alpha(m) = K.$   This convex optimization problem
can be solved using the Karush-Kuhn-Tucker conditions with a Lagrange multiplier for the sum constraint, yielding:
\begin{eqnarray} \label{eqn:alphamin} %
	\alpha(m) &= & \left(  1 -   \frac{c}{(p(m))^{\frac{1}{L-1}}}\right)_+
\end{eqnarray}
where c is chosen so
\begin{eqnarray}
	\nonumber  &\ &\sum_{m=1}^{|M|} \alpha(m) = K\mbox{.}
\end{eqnarray}

A binary bisection algorithm can be used to quickly find $c$ numerically. As a result, given each cache can store $K$ copies of videos, (\ref{eqn:alphamin}) gives the optimal probability any video is present in a typical cache. Any fixed whole storage placement for multiple videos with empirical probability of each video $m$ present in the caches equal to $\alpha(m)$ would serve the maximum expected number of peers for the fixed whole storage placement policy.

A slight additional step can be used to ensure that the total number of videos placed is exactly $K\cdot |H|$. For each video $m$, let $\theta_m$ be the fractional part of $\alpha_m\cdot|H|$. Then, we can take for each video $m$, $C_m=\lceil\alpha_m\cdot|H|\rceil$ with probability $\theta_m$ and $C_m=\lfloor\alpha_m\cdot|H|\rfloor$ with probability $1-\theta_m$, in such a way that $\sum_m C_m \equiv k\cdot |H|$ with probability one, by letting the vector $(C_m:m\in M)$ depend on a uniformly distributed random variable $U$. The value of $p_{miss}$ above can be obtained by averaging over $U$, or $p_{miss}$ no larger than above can be achieved by minimizing $p_{miss}$ with respect to $U$. For fixed whole placement, $p_{miss}$ does not depend on which of the $C_m$ caches video $m$ is placed, for each $m$.

The same performance is obtained when the uniform per-cache storage constraint is replaced by the total storage constraint. This is because the cache-peer connections are random, so the performance does not depend on which cache a video is stored in. 

\subsection{Fixed Fractional Storage Multiple Video Placement Policy}
Recall from the single video analysis for fixed fractional storage placement policy that each cache stores a fractional copy of uniform size, so the marginal performance gain is constant for the placement of any given video up until all caches have a fraction $\frac{1}{L}$ of the videos, and is proportional to the popularity distribution. Applying \autoref{alg:greedywhole}, given each cache can store $K$ copies of videos and each peer is connected to $L$ distinct caches, the placement of a fraction $\frac{1}{L}$ of each of the $K\cdot L$ most popular videos in each cache would serve the maximum number of peers for the fixed fractional storage placement policy. 

The uniform per-cache storage constraint is equivalent to the total cache storage constraint because of the uniform size of fractional copies stored in each cache. 

\subsection{Adaptive Whole Storage Multiple Video Placement Policy}
We see from the single video analysis for the adaptive whole storage placement policy that finding $H_m$ caches that are connected to the maximum number of peers requesting video $m$ is a set cover problem which may require exhaustive search to solve, and the placement problem would be more difficult for multiple videos. Therefore, it is preferable to extend the single video greedy peeling algorithm, \autoref{alg:adaptivewholeyalg}, to a multiple video greedy algorithm, which is in the exact form of \autoref{alg:greedywhole}.

Approximately the same performance can be obtained when the total storage constraint is replaced by the uniform per-cache storage constraint for the following reason. During execution of \autoref{alg:greedywhole} for adaptive whole storage placement, the number of peers served per copy placed decreases monotonically. The number is large for the initial placements because many peers requesting a popular video are connected to the same cache. In the late stages of the algorithm, the numbers of peers served per copy placed typically drops to one or two peers, and there are many possibilities of which cache to use to serve the same maximum number of additional peers. That is because of the large number of low popularity videos. Consequently, the choices of which caches to use in the late stages of \autoref{alg:greedywhole} can be made to balance the individual cache loads, so that individual constraints on the caches are no more restrictive than a single constraint on the total number of copies stored.

\subsection{Adaptive Fractional Storage Multiple Video Placement Policy}
The primal-dual algorithm described in \autoref{sec:adapfrac}, \autoref{alg:adaptivefractionalalg}, has a total cache storage constraint and obtains the optimal placement for a given video $m$. The algorithm can be extended to obtain an optimal placement for multiple videos by combining the primal and dual variables for all videos. For each video $m$, the summation of prices of download rates in step Primal 2 represents the total demand of users for video $m$ in cache $h$. This is similar to finding the marginal performance gain in \autoref{alg:adaptivewholeyalg} and increases the storage of video $m$ in proportion to the marginal performance gain. 

The resulting algorithm differs from the optimal placement algorithm in \cite{HZhang} only by the type of cache storage constraint; the primal-dual algorithm in \cite{HZhang} has a per-cache storage constraint and the primal dual algorithm described above has a total cache storage constraint.  Approximately the same performance can be obtained for the total storage constraint and the uniform per-cache storage constraint, for the same reason as for the adaptive whole storage placement policy. The uniform per-cache and the total cache storage constraints can both be satisfied directly through the primal-dual algorithm by changing the dual variable on cache storages. 

\subsection{Hybrid Multiple Video Placement Policy}%\autoref{alg:greedywhole}
Finally, we construct a hybrid storage placement policy for the whole system with multiple videos. As noted above, the adaptive fractional storage placement policy serves at least as many peers as the other three policies. However, adaptive fractional storage suffers overhead due to the need to encode and decode to provide fractional storage and to adapt to current demand. Hence, we seek another policy serving nearly as many peers, but with less overhead. As observed in \autoref{sec:singlevideocomparison}, for videos with low popularity, the adaptive whole storage placement policy serves nearly as many peers as the adaptive fractional storage placement policy. And for popular videos, the fixed fractional policy serves nearly as many peers as the adaptive fractional storage placement policy. This suggests using a hybrid policy that follows the adaptive whole placement policy for videos with low popularity and the fixed fractional placement policy for popular videos.

From the single video analysis, we have obtained the single video service curves for each video $m\in M$, denoted as $f_{m,0}(C_m)$ and $f_{m,1}(C_m)$ for the fixed fractional and adaptive whole storage placement policies, respectively, which are concave functions of the number of copies stored, $C_m$. For each video $m$, the hybrid storage placement policy chooses to apply one of the two placement policies, and the number of copies $C_m$ to store. The choices are made to maximize the total number of peers served subject to the total storage constraint, $S$. The optimization problem can be stated as follows:
\begin{eqnarray} \label{eqn:hybridoptimizationproblem} %
   	& \max & U({\bf C}, \boldsymbol{\theta})\stackrel{\triangle}{=} \sum_{m\in M} f_{m,\theta (m)}(C_m) \\
	\nonumber & w.r.t. & C_m \mbox{ and } \theta (m),\  m\in M\\
	\nonumber & s.t. &  \sum_{m\in M} C_m \leq S \mbox{ and } \theta (m) \in \{0,1\}\mbox{,}
\end{eqnarray}
where for each video $m$, $\theta(m)$ indicates which placement policy is used.

For $C_m$ given, the choice of $\theta(m)$ is clearly the value $\theta\in \{0,1\}$ that maximizes $f_{m,\theta}(C_m)$. Thus, with $f_m\stackrel{\triangle}{=}\max\left\{f_{m,0}(C_m),f_{m,1}(C_m)\right\}$, \eqref{eqn:hybridoptimizationproblem} is equivalent to the following:
\begin{eqnarray} \label{eqn:hybridoptimizationproblem2} %
   	& \max & U({\bf C})\stackrel{\triangle}{=} \sum_{m\in M} f_{m}(C_m) \\
	\nonumber & w.r.t. & C_m, \  m\in M\\
	\nonumber & s.t. &  \sum_{m\in M} C_m \leq S
\end{eqnarray}

However, \eqref{eqn:hybridoptimizationproblem2} is not a concave optimization problem, because the functions $f_m$ are not concave. As a result, the optimization problem can have local maxima that are not global maxima, making exact solution difficult. To address this problem, we use the concave hull of the objective function, to arrive at the following new optimization problem: 
\begin{eqnarray} \label{eqn:hybridoptimizationproblem3} %
   	& \max & \overline{U}({\bf C})\stackrel{\triangle}{=} \sum_{m\in M} \overline{f}_{m}(C_m) \\
	\nonumber & w.r.t. & C_m, \  m\in M\\
	\nonumber & s.t. &  \sum_{m\in M} C_m \leq S
\end{eqnarray}
where for each video $m$, $\overline{f}_m$ denotes the concave hull of $f_m$.

Let $V^*$ denote the optimal value for \eqref{eqn:hybridoptimizationproblem2} and $\overline{V}$ denote the optimal value for \eqref{eqn:hybridoptimizationproblem3}. We find that for centralized allocation, the effect of the change in the objective function is small. Service curves for fixed fractional and adaptive whole storage policies, $f_{m,0}$ and $f_{m,1}$, for a video $m$ requested by 200 peers (i.e. $|U_m|=200$) are shown as the two solid curves in \autoref{fig:concaveenvelope}. The expected total download rate provided by the caches is plotted versus the number of copies of the video stored in the cache system. The maximum of the two solid curves is $f_m$ and the dotted curve on top is its concave hull, $\overline{f}_m$.

\begin{figure}[htb]
  \begin{center}
 \includegraphics[width=0.53\textwidth,height=0.25\textheight, clip]{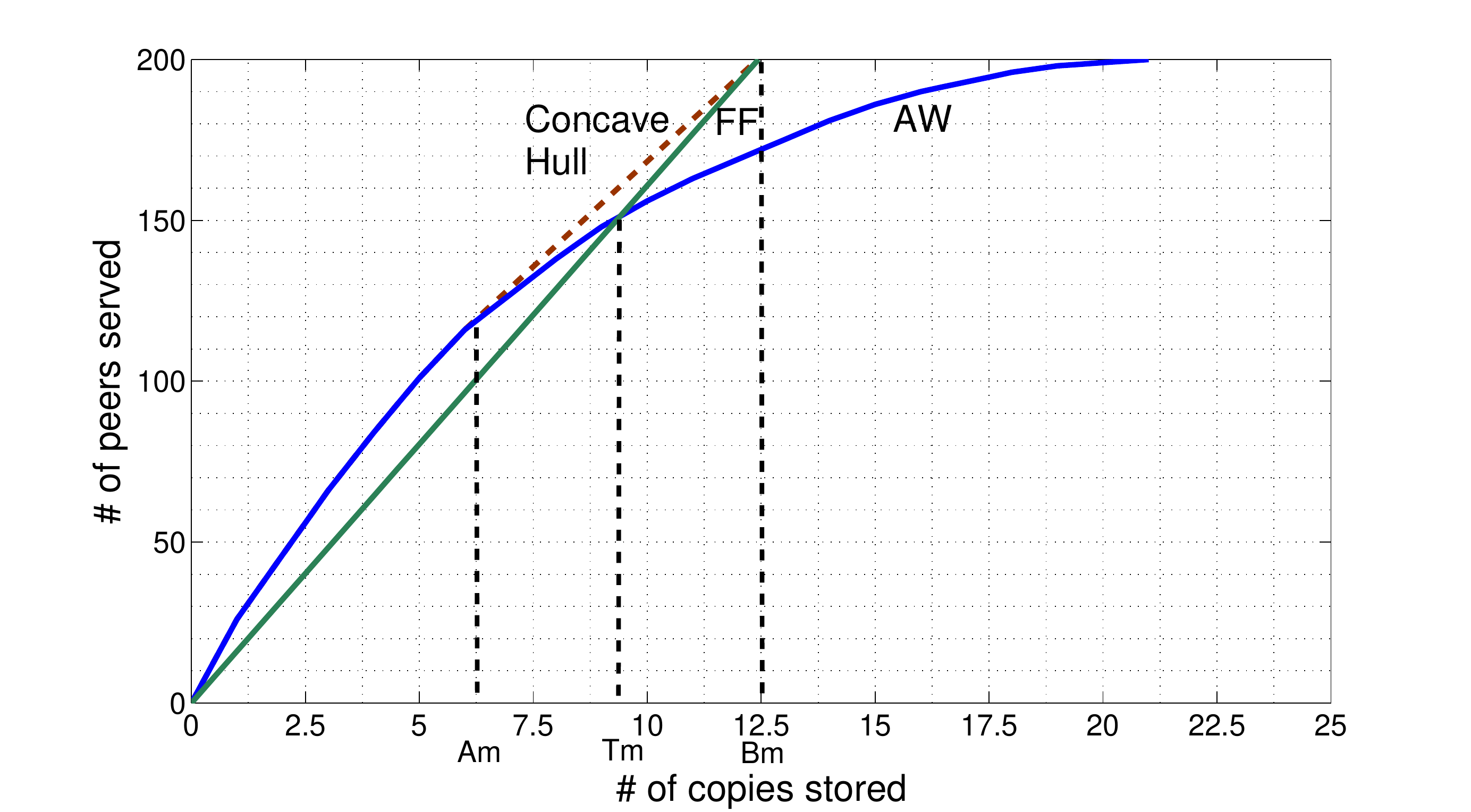}
    \caption{Single video performance of a cache system of 50 caches, 200 peers and four random connections per peer}
    \label{fig:concaveenvelope}
  \end{center}
\vspace{-0.2in}
\end{figure}

There can be at most one crossover between $f_{m,0}$ and $f_{m,1}$ because $f_{m,0}$ is concave and $f_{m,1}$ is linear up to the point all peers requesting video $m$ are served. We observe from \autoref{fig:concaveenvelope} that if a crossover occurs, the concave hull, $\overline{f}_m~$, is linear over an interval $[A_m, B_m]$ and the crossover point $T_m$ falls within the interval. Since $\overline{f}_m$ is a concave upper bound to $f_m~$, $\overline{V}$ is an upper bound on $V^*$ and the solution to \eqref{eqn:hybridoptimizationproblem3} can be obtained by applying \autoref{alg:greedywhole} with a tie breaking rule. The resulting algorithm is denoted as \autoref{alg:hybridalg}.

\begin{algorithm}
  \caption{Hybrid algorithm for multiple video placement with concave hull}
  \label{alg:hybridalg}
\begin{algorithmic}[1]
  \STATE Perform single video fixed fractional and adaptive whole storage placement analysis to obtain the set of placement service curves for all videos, i.e. $f_{m,0}$ and $f_{m,1}$ for all $m\in M$
  \STATE Take the maximum of each pair of service curves to obtain $f_m$
  \STATE Take the concave hull of $f_m$ for each video to obtain $\overline{f}_m$
  \WHILE {storage constraint not violated}
  \STATE - Find the set $V$ of videos $m$ maximizing the marginal performance gain $\overline{f}_m(C_m+1)-\overline{f}_m(C_m)$
  \IF{the video stored in the previous iteration is in $V$} 
  \STATE - Store another copy of that same video
  \ELSE
  \STATE - Store any video in $V$
  \ENDIF
  \ENDWHILE
  \end{algorithmic}
\end{algorithm}

\autoref{alg:hybridalg} determines the number of copies, $\overline{C}_m$ of each video $m$ to store yielding the solution $\overline{\bf C}$ to \eqref{eqn:hybridoptimizationproblem3}. Based on $\bf \overline{C}$ and the definition of $f_m~$, we can find which placement policy to use for each video $m$: $\overline{\theta}(m)=\arg\max_{\theta\in\{0,1\}}\left\{f_{m,\theta}(\overline{C}_m)\right\}$. Thus, $V^*$ can be approximated by $U({\bf \overline{C}}, \overline{\boldsymbol{\theta}})$. We show below that the extent of suboptimality of the placement is small and does not grow with the number of caches or the number of videos. Let $m^*$ denote the video selected in the final iteration of \autoref{alg:hybridalg}, before the storage constraint is met. Define $\Delta(m)\stackrel{\triangle}{=}\max_{C}\left\{\overline{f}_m(C)-f_m(C)\right\}$, which is the maximum difference between $f_m$ and its concave hull, which occurs at $T_m$ if crossover occurs. Note that often $\Delta(m)=0$ for $m$ sufficiently large because if there are not many peers requesting video $m$ then more peers are served by adaptive whole placement for any number of copies $C_m$.

\begin{pro}\label{pro:concaveupperbound} The extent of suboptimality of the placement found by \autoref{alg:hybridalg} is bounded as follows:\\$V^*-U({\bf \overline{C}}, \overline{\boldsymbol{\theta}})\leq\Delta(m^*)$. \end{pro}

\begin{proof}
The solution $\overline{\bf C}$ to \eqref{eqn:hybridoptimizationproblem3} produced by \autoref{alg:hybridalg} has the following property, due to the rule of breaking ties in favor of the video stored in the previous iteration: For any video $m$ with $m\neq m^*$, $\overline{f}_m(\overline{C}_m)=f_m(\overline{C}_m)$, because either there is no crossover, so $\overline{f}_m\equiv f_m$, or there is crossover, but $\overline{C}_m\in[0,A_m]\cup [B_m]$ and $\overline{f}_m = f_m$ on $[0,A_m]\cup[B_m]$. Then we have:
\begin{eqnarray} \label{eqn:conditionalequality} %
   	 f_m(\overline{C}_m)
    \begin{cases}
      = \overline{f}_m(\overline{C}_m), & \mbox{if}\ m\neq m^* \\
      \geq \overline{f}_{m^*}(\overline{C}_{m^*})-\Delta(m^*), & \mbox{if}\ m= m^*
    \end{cases}
\end{eqnarray}
which gives the following inequality by summing over $m$:
\begin{eqnarray} \label{eqn:inequality1}
U({\bf \overline{C}}, \overline{\boldsymbol{\theta}})\geq\overline{U}({\bf \overline{C}}, \overline{\boldsymbol{\theta}})-\Delta(m^*)
\end{eqnarray}
Since the greedy algorithm is optimal for the concave problem \eqref{eqn:hybridoptimizationproblem3} and $\overline{f}_m\ge f_m$ for all $m$, $\overline{U}({\bf \overline{C}}, \overline{\boldsymbol{\theta}})=\overline{V}\geq V^*$. So \eqref{eqn:inequality1} yields the proposition.
\qed
\end{proof}

In the case of a large system presented in the simulation in the next section, $m^*$ is the $400^{th}$ video and the suboptimality gap $\Delta(m^*)=0$ because the adaptive whole storage placement policy performs strictly better than the fixed fractional storage placement policy for the last video, $m^*$, placed. Therefore, $V^*=\overline{V}$ for the simulation. 

\section{Multiple Video Policy Comparisons}\label{sec:multivideocomparison}

To compare the algorithms, we have discussed for a large system, we simulated the algorithms for the system parameters used in \cite{HZhang}: 40,000 peers, 50 caches, and 2000 videos following a 0.8 Zipf popularity distribution, where each peer is connected to 4 random caches selected uniformly, forming a random graph. The system's total storage constraint is 2.5 times the entire video catalog, which is 5000 copies. The placements of multiple videos according to fixed whole storage placement policy, fixed fractional storage placement policy, adaptive whole storage placement policy, and adaptive fractional storage placement policy are shown in \autoref{fig:Simulation_policycopies_example}. The number of video copies stored in the system is plotted versus the videos listed in the order of decreasing popularity. The total numbers of peers served by the cache system are shown in \autoref{table:peersserved} for each multiple video placement policy. The adaptive fractional storage placement policy yields the maximum number of peers that are served by the cache system without help from the server for multiple videos.
\begin{table}[htb]
\caption{Total number of peers served by the cache system}
\centering
\begin{tabular}{c c c}
\hline\hline
Multiple video & Total \# of & Fraction of \\ 
placement policy & peers served & optimal performance \\ [0.5ex]
\hline
Fixed whole & 21747 & 69.9\% \\
Fixed fractional & 26746 & 84.6\%\\
Adaptive whole & 30092 & 95.8\%\\
Adaptive fractional & 31413 & 100\%\\
Hybrid & 31008 & 98.7\%\\
\hline
\end{tabular}
\label{table:peersserved}
\end{table}

For the same large system and total storage constraint, the placement of multiple videos according to the hybrid storage placement policy is shown in \autoref{fig:Simulation_policycopies_example2}. The number of video copies stored in the system is plotted versus the videos in the order of decreasing popularity. Roughly 100 of the most requested videos are stored using the fixed fractional storage placement policy and the remaining videos are using adaptive whole storage placement policy. The placement policy selected for each video by the hybrid storage placement policy is shown in \autoref{fig:Hybrid_choice}. The total number of peers served by the cache system is 31008, which is about 99\% of the performance of the optimal policy - adaptive fractional storage placement policy.

\begin{figure}[htb]
  \begin{center}
 \includegraphics[width=0.53\textwidth,height=0.25\textheight, clip]{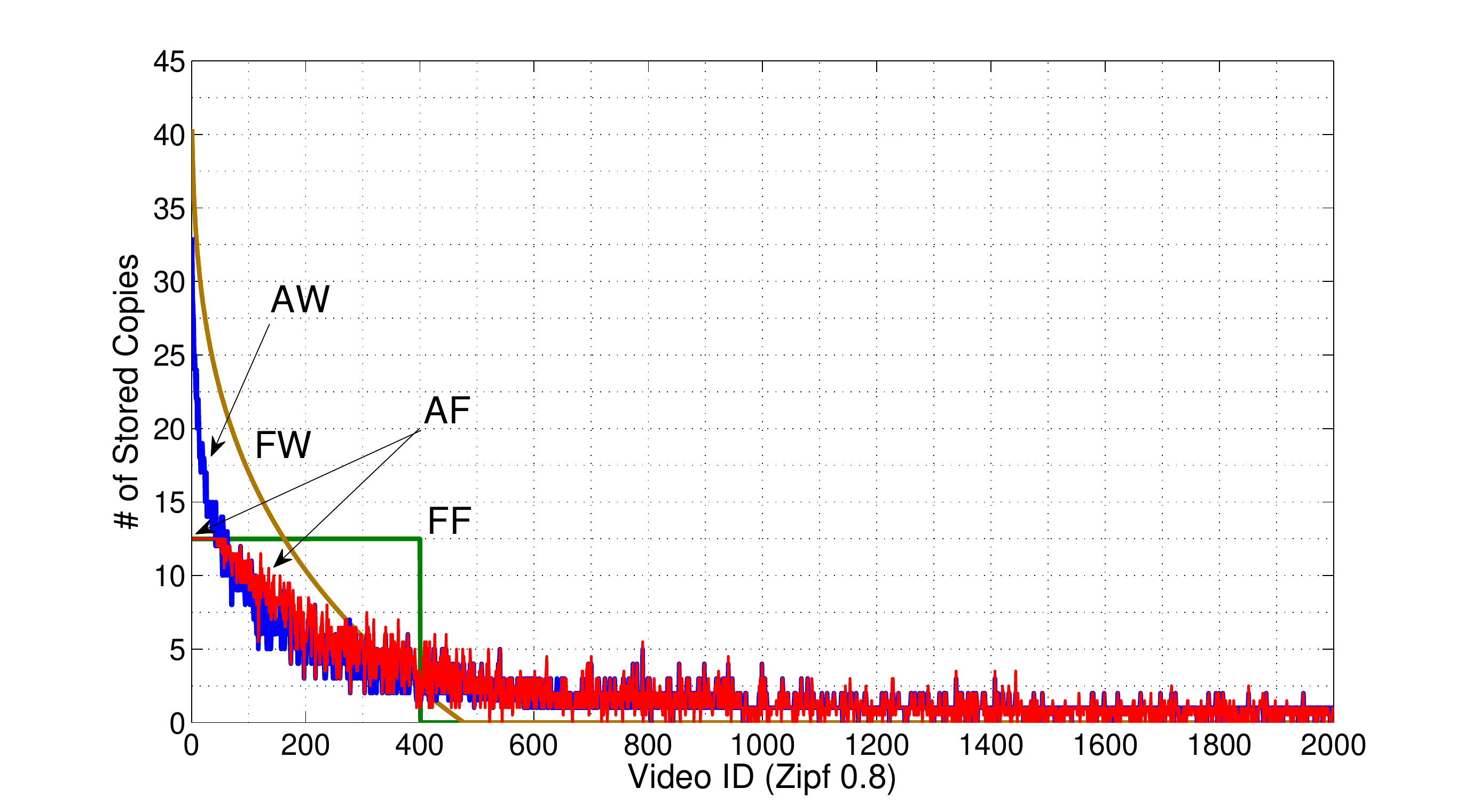}
    \caption{Video copies stored in a whole system of 50 caches, 40,000 peers and four random connections per peer}
    \label{fig:Simulation_policycopies_example}
  \end{center}
\vspace{-0.2in}
\end{figure}

\begin{figure}[htb]
  \begin{center}
 \includegraphics[width=0.53\textwidth,height=0.25\textheight, clip]{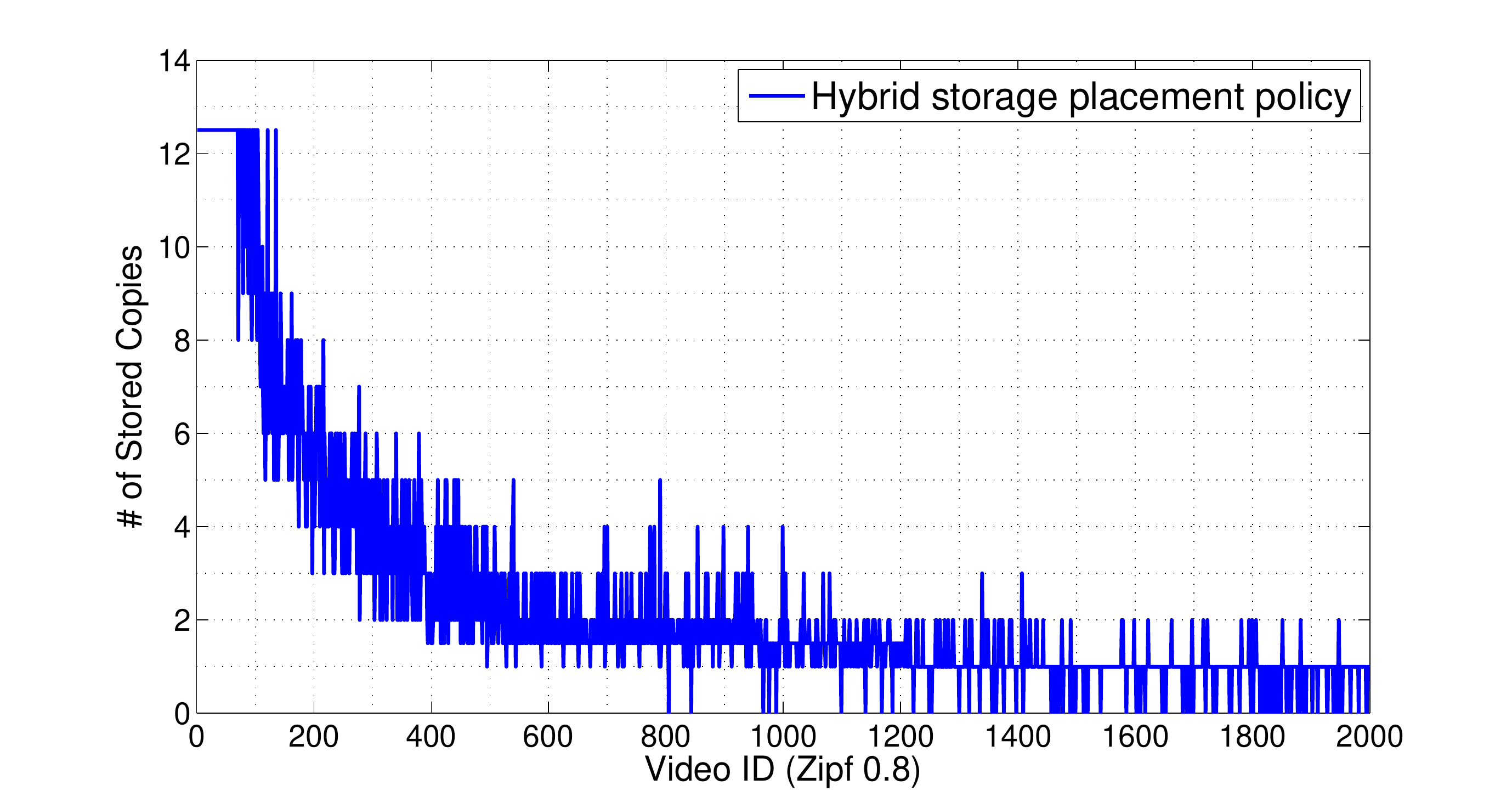}
    \caption{Video copies stored in a whole system of 50 caches, 40,000 peers and four random connections per peer}
    \label{fig:Simulation_policycopies_example2}
  \end{center}
\vspace{-0.2in}
\end{figure}

\begin{figure}[htb]
  \begin{center}
 \includegraphics[width=0.53\textwidth,height=0.25\textheight, clip]{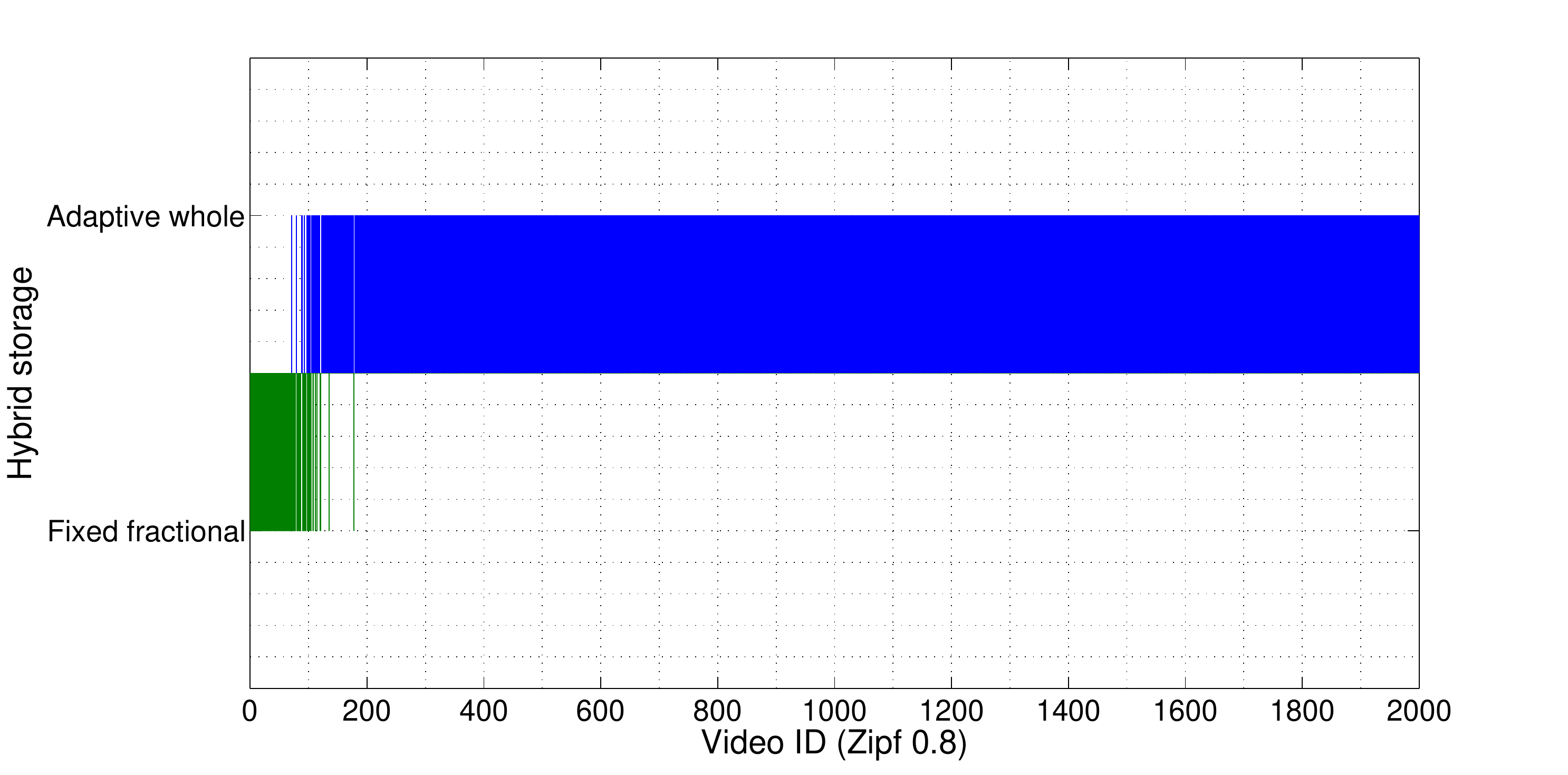}
    \caption{Choice of placement policy in a whole system of 50 caches, 40,000 peers and four random connections per peer}
    \label{fig:Hybrid_choice}
  \end{center}
\vspace{-0.2in}
\end{figure}

\section{Conclusion and Future Work} \label{sec:con}
In this paper, we studied the content placement problem for cache delivery VoD systems. Instead of performing content placement analysis on the whole system with multiple videos, we approached the content placement problem by analyzing the decoupled systems with only a single video. With the insights gained from single video placement policy analysis, we returned to the original content placement problem by constructing a general method for the four placement policies for multiple videos, which places copies of videos so as to maximize the total expected download rate for all users from the caches. We provided analytical and simulation results that answer the key question of how many more peers can be served using fractional storage or adaptive placement. Finally, based on these results, we proposed a hybrid storage placement policy for multiple video placement, which is a lower complexity alternative to the optimal content placement policy serving nearly as many peers.

In practice, \autoref{alg:hybridalg} can be simplified without much loss in performance by estimating in advance which videos to serve using the fixed fraction policy and which to serve using the adaptive whole placement policy. For example, for the system parameters given, we could simply use the fixed fractional policy for the 100 most popular videos. The total number of peers served is not sensitive to he threshold used.

\section{Acknowledgment} \label{sec:ack}
This work was supported by the National Science Foundation under grant no. CCF 10-16959.

\bibliographystyle{IEEEtran}
\bibliography{my-bib}

\end{document}